\documentclass[preprint]{revtex4-2}

\usepackage[utf8]{inputenc}
\usepackage[T1]{fontenc}
\usepackage{amsmath}
\usepackage{amsthm}
\usepackage{amssymb}
\usepackage{hyperref}
\usepackage{dsfont}
\usepackage{graphicx}
\usepackage{tensor}

\hypersetup{colorlinks=true, linkcolor=blue, citecolor=blue}

\newcommand{\Real}{\mathbb{R}}
\newcommand{\Comp}{\mathbb{C}}
\newcommand{\seq}{\subseteq}
\newcommand{\pspace}{\Real^3 \setminus \{0\}}
\newcommand{\Lightcone}{\mathcal{L}_+}
\newcommand{\poincare}{\mathrm{ISO}^+(3,1)}
\newcommand{\trace}{\mathrm{Tr}}
\newcommand{\SO}{\mathrm{SO}}
\newcommand{\SU}{\mathrm{SU}}
\newcommand{\so}{\mathfrak{so}}

\newcommand{\ISO}{\mathrm{ISO}}
\newcommand{\bfk}{\boldsymbol{k}}
\newcommand{\xhat}{\boldsymbol{e}_x}
\newcommand{\zhat}{\boldsymbol{e}_z}
\newcommand{\rp}{\text{Re}}

\newtheorem{theorem}{Theorem}
\newtheorem{lemma}[theorem]{Lemma}

\newtheorem{definition}[theorem]{Definition}

\begin{document}

% Author information
\author{Eric Palmerduca}
\email{ep11@princeton.edu}
\affiliation{Department of Astrophysical Sciences, Princeton University, Princeton, New Jersey 08544}
\affiliation{Plasma Physics Laboratory, Princeton University, Princeton, NJ 08543,
U.S.A}

\author{Hong Qin}
\email{hongqin@princeton.edu}
\affiliation{Department of Astrophysical Sciences, Princeton University, Princeton, New Jersey 08544}
\affiliation{Plasma Physics Laboratory, Princeton University, Princeton, NJ 08543,
U.S.A}

% Date/Title
\title{Graviton topology}
\date{\today}

%%%%%%%%%%%%%%%%%% Abstract %%%%%%%%%%%%%%
\begin{abstract}
Over the past three decades, it has been shown that discrete and continuous media can support topologically nontrivial waves. Recently, it was shown that the same is true of the vacuum, in particular, right (R) and left (L) circularly polarized photons are topologically nontrivial. Here, we study the topology of another class of massless particles, namely gravitons. We show that the collection of all gravitons forms a topologically trivial vector bundle over the lightcone, allowing us to construct a globally smooth basis for gravitons. The graviton bundle also has a natural geometric splitting into two topologically nontrivial subbundles, consisting of the R and L gravitons. The R and L gravitons are unitary irreducible bundle representations of the Poincar\'{e} group, and are thus elementary particles; their topology is characterized by the Chern numbers $\mp 4$. This nontrivial topology obstructs the splitting of graviton angular momentum into spin and orbital angular momentum.\end{abstract}

\maketitle

\section{Introduction}
It has been established that topologically nontrivial modes can exist in various media such as solids ~\cite{Qi2011}, fluids \cite{Delplace2017,Tauber2019,Souslov2019,Perrot2019,Faure2022}, and plasmas \cite{Yang2016,Gao2016,Parker2020a,Fu2021,Fu2022,Qin2023, Qin2024plasma}. The study of such modes has lead to important physical discoveries such as the quantum Hall effects and the bulk-edge correspondence. More recently it has been shown that the vacuum also supports topologically nontrivial modes, in particular, left (L) and right (R) circularly polarized photons form topologically nontrivial vector bundles with Chern numbers $+2$ and $-2$, respectively \cite{PalmerducaQin_PT}. More generally, it is possible for elementary massless particles to possess nontrivial topology due to a hole in their momentum space at $k=0$, although only the photon case has been studied in depth. Here we present a thorough and rigorous examination of the topology of another massless particle, the (linearized) graviton. While gravitational waves, and their quanta, gravitons, were predicted early in the development of general relativity and quantum field theory, they are of particular interest now due to two recent experimental breakthroughs. Within the past decade, LIGO has detected gravitational waves from black hole and neutron star merger events, experimentally confirming the existence of gravitational waves \cite{GravWave2016Blackhole,GravWave2017NeutronStar}. Even more recently, graviton-like quasiparticles have been observed in fractional quantum Hall liquids \cite{Liang2024}. As topology, particularly Chern numbers, plays a fundamental role in the quantum Hall effects \cite{Thouless1982,Avron1983,Kohomoto1985,Hatsugai1993,Avron2003}, an understanding of the topology of gravitons is needed.

Working in the transverse-traceless (TT) gauge for weak gravity \cite{Schutz2009}, we show that the collection of all gravitons forms a rank-2 vector bundle $\Xi$ over the the forward lightcone $\Lightcone$. We prove that the graviton bundle $\Xi$ is topologically trivial, that is, 
\begin{equation}\label{eq:trivial_splitting}
    \Xi \cong \tau_1 \oplus \tau_2
\end{equation}
where $\tau_1$ and $\tau_2$ are topologically trivial line bundles over $\Lightcone$. As such, the Chern number of $\Xi$ is zero. This implies that it is possible to write down a smooth global basis for gravitational waves in momentum space, which we do using a modification of the clutching construction from algebraic topology \cite{HatcherVBKT,PalmerducaQin_PT}.  While such global bases are frequently invoked in the literature \cite{Yu1999, Maggiore2005, Hu2021, Carney2024}, until now there was no proof of their existence. In fact, the stronger assumption is typically made that there exists a smooth global basis which is linearly polarized \cite{Yu1999, Maggiore2005, Hu2021, Carney2024}. This stronger assumption is, however, false; we prove that $\Xi$ has no linearly polarized subbundles, even topologically nontrivial ones. Thus, any global basis necessarily involves elliptical polarizations.

While the total graviton bundle $\Xi$ is topologically trivial and can be broken down as in (\ref{eq:trivial_splitting}), its Poincar\'{e} symmetry gives a more natural geometric splitting into the $R$ and $L$ graviton bundles $\Xi_+$ and $\Xi_-$. These consist of gravitons with helicities $+2$ and $-2$, respectively, describing their behavior under rotations about their momentum axis. $\tau_1$ and $\tau_2$ are topologically trivial while linear polarizations do not even form well-defined vector bundles. In contrast, $\Xi_+$ and $\Xi_-$ are well-defined vector bundles with nontrivial topology completely characterized by the Chern numbers $\mp 4$.  This topological nontriviality reflects the fact that the internal (base manifold) and external (fiber) degrees of freedom (DOFs) are twisted together. This twisting acts as a topological obstruction to the splitting of graviton angular momentum into spin angular momentum (SAM) and orbital angular momentum (OAM). This is important since there have been recent efforts to study the OAM of gravitational waves  predicated on the validity of such a splitting \cite{Bialynicki-Birula_2016, Baral2020, Wu2023}.

This article is organized as follows. In Section \ref{sec:graviton_bundle} we show that in the limit of weak gravity, the collection of all gravitons form a smooth vector bundle $\Xi$ over the forward lightcone. In Section \ref{sec:Poincare_symmetry}, we show that the Poincar\'{e} symmetry of $\Xi$ naturally decomposes the bundle into the R and L gravitons $\Xi_\pm$, which correspond to elementary particles. In Section \ref{sec:topology_of_gravitons}, we completely classify these vector bundles via their Chern numbers, showing that $\Xi$ is trivial while $\Xi_{\pm}$ are nontrivial. We proceed in Section \ref{sec:global_basis} to use a modified clutching construction to explicitly construct a smooth global basis of $\Xi$. In Section \ref{sec:linear_polarizations}, we show that unlike the trivial and circular polarizations, linear polarizations do not form well-defined vector bundles and thus have no associated topology. In Section \ref{sec:OAM_SAM} it is shown that the topological nontriviality of $\Xi_\pm$ obstructs the splitting of graviton angular momentum into OAM and SAM. Lastly, in Section \ref{sec:curved_spacetime} we discuss extensions of these results to gravitational waves in curved background spacetimes.

%%%%%%%%%%%%% The graviton bundle %%%%%%%%%%%%%%%%%
\section{The graviton bundle $\Xi$}\label{sec:graviton_bundle}

We work in the limit of weak gravity, in which the spacetime metric $g_{\mu \nu}$ can be expanded around the flat Minkowski metric $\eta_{\mu \nu}$ with signature $(-,+,+,+)$ as
\begin{equation} \label{eq:flat_assumption}
    g_{\mu \nu} = \eta_{\mu \nu} + h_{\mu \nu},
\end{equation}
where $h_{\mu \nu}$ is symmetric and $|h_{\mu \nu}| \ll 1$ \cite{Schutz2009}. $h_{\mu \nu}$ is the graviton field \cite{Zee2010}. Following Schutz \cite{Schutz2009}, in vacuum, far from the sources of the field, the Einstein equations in the TT gauge take the form
\begin{equation}
    \eta^{\mu \nu}\tensor{h}{^{\alpha\beta}_{,\mu\nu}} = 0.
\end{equation}
The solutions can be written in Fourier space as
\begin{equation}
    h^{\alpha \beta} = A^{\alpha \beta}\exp(ik^{\mu} x_{\mu})
\end{equation}
where $k^\mu = (\omega, \bfk)$ is the four-momentum and $A^{\alpha \beta}$ is a complex symmetric tensor. We refer to these solutions as gravitons, as they are the modes of the graviton field.  $k$ and $A$ are both subject to constraints. The Einstein equation requires that
\begin{equation}
    k^\alpha k_\alpha = 0,
\end{equation}
so $\omega = \pm|\bfk|$, i.e., $k$ resides on the lightcone. It is actually sufficient to consider only the $\omega = |\bfk|$ solutions in which $k$ is restricted to the forward lightcone $\Lightcone$, as the $\omega = -|\bfk|$ solutions are then determined by the condition that $h^{\mu \nu}$ is real. Notably, the forward $\Lightcone$ has a hole at the origin $k^\mu = 0$, as can be seen by a number of arguments. From a physical perspective, modes with $k^\mu = 0$ are not propagating and thus do not represent gravitons. Such modes are not accessible by Lorentz transforming propagating modes since massless particles have no rest frame. Mathematically, the origin must be removed in order for $\Lightcone$ to be a regular submanifold of Minkowski space. Furthermore, as pointed out by Staruszkiewicz \cite{Staruszkiewicz1973b}, the forward lightcone is only geodesically complete if $k\neq 0$ is excluded. By parametrizing $\Lightcone$ by its spatial part $\bfk$, the homeomorphism $\Lightcone \cong \pspace$ is obtained.

Beyond the constraint that $A^{\alpha \beta}$ is symmetric, it must additionally satisfy the TT gauge conditions \cite{Schutz2009}:
\begin{subequations}
\begin{gather}
    A^{\alpha \beta}k_{\beta} = 0 \label{eq:transverse_gauge_full}\\
    \tensor{A}{^\alpha_\alpha} = 0 \label{eq:traceless_gauge_full}\\
    A^{\alpha \beta}U_\beta = 0, \label{eq:temporal_gauge}
\end{gather}
\end{subequations}
where $U_\beta$ is any fixed timelike unit vector. The first two conditions reflect the transverse and traceless gauge choices, respectively. The third condition fixes a residual temporal gauge freedom. The canonical choice, which we adopt here, is to take $U_\beta = \delta_{\beta}^0$ since this is the distinguished timelike unit vector. The temporal gauge condition thus imposes that
\begin{equation}\label{eq:complete_gauge}
    A^{\alpha 0} = A^{0 \alpha} = 0,
\end{equation}
allowing us to consider just the tensors $A^{ij}$, where $i,j \in \{1,2,3\}$ span the spatial DOFs. The gauge conditions (\ref{eq:transverse_gauge_full}) and (\ref{eq:traceless_gauge_full}) then reduce to
\begin{subequations}
\label{eq:gauge_constraints}
\begin{gather}
    A^{ij}k_j = 0 \label{eq:transverse_gauge}\\
    \tensor{A}{^i_i} = 0 \label{eq:traceless_gauge}.
\end{gather}
\end{subequations}

Thus, the set of all gravitons is given by
\begin{equation}
    \Xi = \{(k, A^{ij})\}
\end{equation}
with $k \in \Lightcone$ and $A^{ij}$ symmetric and subject to the constraints (\ref{eq:gauge_constraints}). Denote by $\Xi(k_0)$ the fiber at $k_0$ consisting of the elements of $\Xi$ with $k = k_0$. The transverse gauge condition (\ref{eq:transverse_gauge}) imposes three independent linear constraints on the five-dimensional complex vector space $V_3$ of traceless symmetric rank-2 tensors on $\Comp^3$, and thus each fiber $\Xi(k_0)$ is a two-dimensional complex vector space. Furthermore, the transverse condition is smooth in $k$, and thus the fibers fit together to give $\Xi$ the form of a rank-2 complex vector bundle over $\Lightcone \cong \pspace$; we call $\Xi$ the graviton bundle. As vector bundles are topological objects, this formalism gives a concrete way to study the topology of gravitons. The essential problem is to determine the topology of $\Xi$ and its subbundles.

The most obvious topological question is the following: is $\Xi$ topologically trivial, that is, does $\Xi$ possess a global basis? It is always possible to form a local basis. For example, the plus-cross tensors
\begin{subequations}
\label{eq:B_z}
\begin{align}
    B_{z,+} =
    &\begin{pmatrix}
        1 & 0 & 0 \\
        0 & -1 & 0 \\
        0 & 0 & 0
    \end{pmatrix} \\
    B_{z,\times} = 
    &\begin{pmatrix}
        0 & 1 & 0 \\
        1 & 0 & 0 \\
        0 & 0 & 0
    \end{pmatrix},
\end{align}
\end{subequations}
are a commonly used basis at $\bfk_0 = \boldsymbol{e}_z = (0,0,1)$, and it is straightforward to smoothly extend this locally to a small neighborhood about $\boldsymbol{e}_z$ (e.g., by parallel transport via any connection). However, it is commonly assumed that this can be smoothly extended to a global basis \cite{Yu1999, Maggiore2005, Hu2021, Carney2024}. One explicit globalized choice which has appeared in the literature is \cite{Yu1999,Maggiore2005}
\begin{subequations}
\label{eq:B_plus_cross_Maggiore}
\begin{align}
    B_+(\bfk) &= \boldsymbol{u}(\bfk) \otimes \boldsymbol{u}(\bfk) - \boldsymbol{v}(\bfk)\otimes \boldsymbol{v}(\bfk) \label{eq:B_plus_Maggiore}\\
    B_\times(\bfk) &= \boldsymbol{u}(\bfk) \otimes \boldsymbol{v}(\bfk) + \boldsymbol{v}(\bfk)\otimes \boldsymbol{u}(\bfk) \label{eq:B_cross_Maggiore},
\end{align}
\end{subequations}
where $(\boldsymbol{u}(\bfk), \boldsymbol{u}(\bfk), \hat{\bfk})$ is supposed to be some orthonormal basis of $\Real^3$ at each $\bfk$. The issue is that it is not actually possible to smoothly choose such an orthonormal basis! Indeed, when restricted to the unit sphere $S^2$, $\boldsymbol{u}(\bfk)$ and $\boldsymbol{v}(\bfk)$ would be continuous nonvanishing tangent vector fields on $S^2$, violating the hairy ball theorem (\cite{Milnor1978}, Theorem 1). Thus, any basis $(B_+(\bfk),B_\times(\bfk))$ of the form (\ref{eq:B_plus_cross_Maggiore}) has singularities.

The reason why it is not obvious that \emph{some} smooth global basis for gravitons exists is that gravitons are massless. Indeed, for particles with mass $m > 0$, the momentum space is not the forward lightcone, but rather a mass hyperboloid
\begin{equation}
    M_m = \{k \in \Real^4 | k^\mu k_\mu = -m\} \cong \Real^3.
\end{equation}
Massive particles are thus described by vector bundles over the base manifold $M_m \cong \Real^3$ which is contractible to a point. Any vector bundle over a contractible manifold is a trivial bundle (\cite{HatcherVBKT}, Corollary 1.8), and therefore massive particles are topologically trivial and possess a global basis. However, for massless particles such as photons and gravitons, the momentum space is the topologically nontrivial $\Lightcone$; it cannot be contracted to a point due to the hole at $k = 0$. As such, there exist nontrivial bundles over $\Lightcone$. Indeed, the $R$ and $L$ circularly polarized photon bundles $\gamma_\pm$ were shown to be nontrivial \cite{PalmerducaQin_PT}. Similarly, we will show that $\Xi$ can be broken into two topologically nontrivial $R$ and $L$ circularly polarized subbundles $\Xi_\pm$, although they are topologically distinct from $\gamma_\pm$. Despite the nontriviality of $\Xi_\pm$, the total graviton bundle $\Xi$ actually is trivial, and we will construct a globally smooth basis for it, although it is more complicated than that described in (\ref{eq:B_plus_cross_Maggiore}).

%%%%%%%%%% Decomposition of the graviton bundle into elementary particles %%%%%%%%%%%%%%%%%%%%%%%%%%

\section{Decomposition of the graviton bundle into elementary particles}\label{sec:Poincare_symmetry}
Although elementary particles are conventionally defined as unitary irreducible representations (UIRs) of the proper orthochronous Poincar\'{e} group $\poincare$ on a Hilbert space \cite{Wigner1939, Weinberg1995}, they can be equivalently defined as UIRs of $\ISO^+(3,1)$ on vector bundles \cite{Simms1968, Asorey1985, PalmerducaQin_PT}. Representations on vector bundles are the natural generalization of vector space representations, and describe symmetries of the underlying system. In particular, if E is a vector bundle over the base manifold $M$, then $E$ is a representation of the $\poincare$ with action $\Sigma$ if for each $L \in \poincare$, $\Sigma(L)$ is a linear isomorphism between the fibers $E(k)$ to $E(L k)$. Furthermore, it is required $\Sigma$ respect the group structure, that is, $\Sigma(L_1 L_2) = \Sigma(L_1)\Sigma(L_2)$ for any $L_1, L_2 \in \poincare$. Massless particles correspond to representations on bundles over the lightcone and massive particles to bundles over mass hyperboloids. Unlike massive particles which can correspond to higher rank (trivial) bundles, massless particles are always line bundles and can be nontrivial \cite{PalmerducaQin_PT}. We will show that $\Xi$ is a unitary vector bundle representation of $\poincare$, and find its decomposition into UIRs, that is, into elementary particles. 

The group $\ISO^+(3,1)$ consists of Lorentz transformations $\Lambda$, spacetime translations $a^\mu \in \mathbb{R}^4$, and compositions of the two. We define the action of $a$ on $(k,A^{\mu \nu}) \in \Xi$ by
\begin{equation}\label{eq:spacetime_action}
    \Sigma(a)(k,A^{\alpha \beta}) = (k, e^{ia^\mu k_\mu}A^{\alpha \beta}),
\end{equation}
reflecting the effect of spacetime translations in momentum space. Ideally, for a general Lorentz transformation $\Lambda$, $(k,A^{\alpha \beta})$ would transform covariantly:
\begin{equation}\label{eq:covariant_transform}
    (k,A^{\alpha \beta}) \mapsto (k', A'^{\alpha \beta})\doteq (\Lambda k, \tensor{\Lambda}{^\alpha_\mu}\tensor{\Lambda}{^\beta_\nu}A^{\mu \nu}).
\end{equation}
$(k', A')$ satisfies the transverse and traceless gauge conditions given by equations (\ref{eq:transverse_gauge_full}) and (\ref{eq:traceless_gauge_full}). However, when $\Lambda \not\in \SO(3)$, $A'$ generally has nonvanishing temporal components, violating the temporal gauge condition (\ref{eq:complete_gauge}) and implying that $(k', A') \not \in \Xi$. It is, however, always possible to find a gauge transformation which takes $A'$ to a physically equivalent $A''$ such that $A''$ satisfies the TT gauge constraints (\ref{eq:complete_gauge}) and (\ref{eq:gauge_constraints}) so that $(k',A'') \in \Xi$. Indeed, any $C^\alpha$ satisfying $C^\alpha k'_\alpha=0$ induces the gauge transformation \cite{Schutz2009}
\begin{equation}\label{eq:gauge_transform}
    A''^{\alpha \beta} = A'^{\alpha \beta} - iC^\alpha k'^\beta - i k'^\alpha C^\beta.
\end{equation}
Since $k'^\alpha k'_\alpha = C^\alpha k'_\alpha = 0$, $A''$ is still transverse and traceless:
\begin{subequations}
\begin{gather}
    A''^{\alpha \beta} k_\beta = 0 \\
    \tensor{{A''}}{^\alpha_\alpha} = 0.
\end{gather}
\end{subequations}
Imposing the temporal gauge condition $A''^{\alpha 0} = 0$ fixes $C^\alpha$:
\begin{equation}
    C^\alpha = \frac{i}{k'^0}\Big(-A'^{\alpha 0} + \frac{A'^{00}}{2k'^{0}} k'^\alpha \Big).
\end{equation}
Then $(k',A'') \in \Xi$, so the action of $\Lambda$ on $\Xi$ is given by
\begin{equation}\label{eq:Lorentz_action}
    \Sigma(\Lambda)(k,A) \rightarrow (\Lambda k, A'').
\end{equation}
We note that the action of a rotation $R \in \SO(3) \seq \poincare$ on $(k,A) \in \Xi(k)$ is simply given by
\begin{equation}\label{eq:rotation_action}
\Sigma(R)(k,A) = (Rk, RAR^{-1}),
\end{equation}
where $A$ and $R$ are both considered as three-by-three matrices.

To show that this action is unitary, we must first define a Hermitian product on $\Xi$. Consider the vector space $V_4$ consisting of symmetric, traceless rank-2 tensors on $\Comp^4$ which has the indefinite Hilbert-Schmidt product given by
\begin{align}\label{eq:indefinite_HS}
    \langle A^{\alpha \beta}, B^{\alpha \beta} \rangle = \frac{1}{2}(A^{*})^{\alpha \beta} B_{\beta \alpha}.
\end{align}
When restricted to $V_3$, the subspace of $V_4$ with $A^{\alpha 0 } = 0$, this product is positive-definite making $V_3$ into a Hilbert space. On $V_3$, the Hilbert-Schmidt product takes the form
\begin{equation}
    \langle A, B \rangle = \frac{1}{2}\trace(A^* B)
\end{equation}
By equipping every fiber of $\Xi$ with this inner product, $\Xi$ is made into Hermitian vector bundle. It is easy to see that the Poincar\'{e} action $\Sigma$ is unitary with respect to this product, that is, for $A,B \in \Xi(k)$:
\begin{equation}
    \langle \Sigma(\Lambda)A, \Sigma(\Lambda)B \rangle = \langle A,B \rangle.
\end{equation}
In particular, the inner product of $A$ and $B$ is a Lorentz scalar by (\ref{eq:indefinite_HS}), so it is preserved by the covariant transformation (\ref{eq:covariant_transform}). The subsequent gauge transformation (\ref{eq:gauge_transform}) preserves the product since $k'^\alpha k'_\alpha = C^\alpha k'_\alpha = 0$. We have thus proved the following result:
\begin{theorem}\label{thm:Xi_representation}
    $\Xi$ is a unitary vector bundle representation of $\poincare$ with action $\Sigma$ given by equations (\ref{eq:spacetime_action}) and (\ref{eq:Lorentz_action}).
\end{theorem}

This allows us to use the following theorem, proved in Ref. \cite{PalmerducaQin_PT}, to decompose $\Xi$ into UIRs. For $\bfk \in \Lightcone$, denote by $R_{\bfk}(\psi) \in \SO(3) \seq \poincare$, the rotation by angle $\psi$ about the $\bfk$ axis.
\begin{theorem}\label{thm:bundle_decomp}
    Let $\pi:E\rightarrow \mathcal{L}_+$ be a unitary vector bundle representation of $\mathrm{ISO}^+(3,1)$ of rank $r$ such that a spacetime translation by $a \in \mathbb{R}^4$ acts on $(k,v) \in E_k$ by $\Sigma_{a}(k,v) = (k,e^{ik^\mu a_\mu}v)$. Then $E$ decomposes as
    \begin{equation}\label{eq:E_decomp}
        E = E_1 \oplus \cdots \oplus E_r,
    \end{equation}
    where the $E_j$ are unitary irreducible line subbundles of $E$. Each $E_j$ has helicity $h_j$, meaning that for any $(k,v_j) \in E_j$, 
    \begin{equation}
        \Sigma(R_{\bfk}(\psi)) (k,v_j) = e^{-ih_j \psi}(k,v_j).
    \end{equation}
    If the $h_j$ are distinct, the decomposition (\ref{eq:E_decomp}) is unique.
\end{theorem}

\begin{theorem}[R and L circularly polarized subbundles]\label{thm:RL_subbundles}
    The graviton bundle $\Xi$ decomposes as
    \begin{equation}
        \Xi = \Xi_+ \oplus \Xi_- \label{eq:RL_decomp}
    \end{equation}
    where $\Xi_\pm$ are unique UIR line bundle representations of $\poincare$ with helicities $\pm 2$. They are thus elementary particles. Elements of $A \in \Xi_\pm(k)$ have the form
    \begin{equation}\label{eq:real_decomp}
        A = c(B_1 + iB_2)
    \end{equation}
    where $c \in \Comp$ and $B_1,B_2 \in \Xi(k)$ are real and orthonormal. We call $\Xi_+$ and $\Xi_-$ the R and L circularly polarized bundles in analogy with the R and L photons.
\end{theorem}
\begin{proof}
    By Theorems \ref{thm:Xi_representation} and \ref{thm:bundle_decomp}, $\Xi$ decomposes into two bundles
    \begin{equation}
        \Xi = \Xi_+ \oplus \Xi_-
    \end{equation}
    with helicities $h_+$ and $h_-$, respectively. To show that $h_\pm = \pm 2$, it suffices to consider any fixed fiber, say $\Xi(\zhat)$. The elements 
    \begin{subequations}
    \begin{align}
        A_{z,\pm} &= \frac{1}{\sqrt{2}}(B_{z,+} \pm i B_{z,-}) \\
        &= \frac{1}{\sqrt{2}}
        \begin{pmatrix}
            1 & \pm i & 0 \\
            \pm i & -1 & 0 \\
            0 & 0 & 0
        \end{pmatrix}
    \end{align}
    \end{subequations}
    of the fiber $\Xi(\boldsymbol{e}_z)$ satisfy
    \begin{equation}
        \Sigma(R_{\boldsymbol{e}_z}(\psi))A_{z,\pm} =  e^{\mp2i\psi}A_{z,\pm},
    \end{equation}
    for any $\psi$, showing that $h_\pm = \pm 2$. 
    
    Now let $L_{\bfk}$ be some Lorentz transformation taking $\zhat$ to $\bfk$. Since $\Xi_\pm$ are $\poincare$ symmetric,
    \begin{equation}
        \Sigma(L_{\bfk})A_{z,\pm} = \frac{1}{\sqrt{2}}\big(\Sigma(L_{\bfk})B_{z,+} \pm i \Sigma(L_{\bfk})B_{z,\times} \big)
    \end{equation}
    span the fibers $\Xi_\pm(\bfk)$. The terms $\Sigma(L_{\bfk})B_{z,+}$ and $\Sigma(L_{\bfk})B_{z,\times}$ are real and orthonormal, so the elements of $\Xi_\pm$ have the form (\ref{eq:real_decomp}).
\end{proof}

%%%%%%%%%% The topology of the total and circularly polarized graviton bundles %%%%%%%%%%%%%%%%%%%%%%%%%%%%%%%%%%
\section{The topology of the total and circularly polarized graviton bundles} \label{sec:topology_of_gravitons}

In this section we characterize the topology of $\Xi$ and $\Xi_\pm$. While we mostly use the shorthand $\Xi$ for the graviton bundle, technically the bundle is $\pi: \Xi \rightarrow \Lightcone$ where $\pi$ is the projection $(k,A) \mapsto k$ onto the base manifold. It is often more convenient to work with the bundle obtained by restricting the base manifold to the unit sphere $S^2 \seq \pspace \cong \Lightcone$. Since $\Lightcone$ deformation retracts onto $S^2$, bundles over $\Lightcone$ are functorially isomorphic to those over $S^2$ (\cite{HatcherVBKT}, Corollary 1.8). In particular, if $j:S^2 \hookrightarrow \Lightcone$ is the inclusion and $r:\Lightcone \rightarrow S^2$ is the retraction $\bfk \mapsto \hat{\bfk}$, then the pullback map $j^*$ sends bundles over $\Lightcone$ to bundles over $S^2$ and has inverse $r^*$. If $E$ is a bundle over $\Lightcone$, its restriction to $S^2$ is given by the pullback bundle $E|_{S^2} = j^*E$. Since $S^2$ is even-dimensional and compact, complex bundles over $S^2$ possess Chern numbers as topological invariants \cite{Tu2017differential}. Bundles over odd-dimensional or noncompact manifolds do not generally possess Chern numbers, but the isomorphism $j^*$ allows Chern numbers to be assigned to bundles over $\Lightcone$ such as $\Xi$. For example, the first Chern number of $\Xi$ is given by
\begin{equation}
    C_1(\Xi) \doteq  C_1(j^* \Xi) = C_1(\Xi|_{S^2}).
\end{equation}

Complex vector bundles over $S^2$ with rank $1$ or $2$ are completely classified up to isomorphism by their first Chern number $C_1$ \cite{McDuff2017,PalmerducaQin_PT}, and thus the same is true of such bundles over $\Lightcone$ by the equivalence induced by $r^*$ and $j^*$. This allows us to fully classify the graviton bundles.
\begin{theorem}[The topology of gravitons]\label{thm:graviton_topology}
    The total graviton bundle $\Xi$ is topologically trivial, while the circularly polarized bundles $\Xi_\pm$ are topologically nontrivial with Chern numbers $\mp 4$. 
\end{theorem}
\begin{proof}
    We will determine $C_1({\Xi}_\pm)$ via the Berry connection induced by the Hilbert-Schmidt product. To do so, we need smooth bases of the restricted bundles $\Xi_\pm|_{S^2}$. It is too much to ask for global bases of $\Xi_\pm|_{S^2}$ since we will find that these bundles are topologically nontrivial. However, we can find smooth bases almost everywhere and this is sufficient to calculate the Chern number. We will obtain such a basis by applying rotations to $A_{z,\pm}$. In particular, define
    \begin{equation}
    A_\pm(\theta,\phi) = \Sigma\Big(R_{\zhat}(\phi)R_{\xhat}(\theta)\Big)A_{z,\pm}
    \end{equation}
    for $\theta \in (0,\pi)$ and $\phi \in [0,2\pi)$. Then with $$\hat{\bfk} = (\sin \theta \cos \phi, \sin \theta \sin \phi, \cos \theta),$$ we have that
    \begin{equation}
        A_\pm(\theta, \phi) \in \Xi_\pm(\hat{\bfk})
    \end{equation}
    except at the poles where $A_\pm(\theta, \phi)$ would become discontinuous. We can thus express the Berry connections on $\Xi_\pm|_{S^2}$ almost everywhere in these coordinates \cite{Frankel2011}: 
\begin{subequations}
\begin{align}
    \sigma_\pm(\theta,\phi) &= \langle A_\pm(\theta,\phi), dA_\pm(\theta,\phi)\rangle \\
    &= \mp 2 i \cos \theta d\phi
\end{align}
\end{subequations}
The Berry curvatures are then
\begin{subequations}
\begin{align}
    \Omega_\pm &= d\sigma_\pm + \sigma_\pm \wedge \sigma_\pm \\
    &= \pm 2i \sin \theta d\theta \wedge d\phi \\
    &= \pm 2 i \, dS \label{eq:Berry_curvature}.
\end{align}
\end{subequations}
where $dS$ is the spherical area element of $S^2$. Note that while our calculation excluded the poles, the Berry curvature is a globally defined and smooth 2-form \cite{Tu2017differential}, and thus (\ref{eq:Berry_curvature}) must hold globally. The Chern numbers can then be calculated \cite{Tu2017differential}:
\begin{subequations}
\begin{align}
    C_1(\Xi_\pm) &= C_1(\Xi_\pm|_{S^2}) \\
    &= \frac{i}{2\pi} \int_{S^2} \Omega_\pm \\
    &= \mp 4.
\end{align}
\end{subequations}
The first Chern number is additive, so by (\ref{eq:RL_decomp})
\begin{equation}
    C_1(\Xi) = C_1(\Xi_+) + C_1(\Xi_-) = 0.
\end{equation}
Since the trivial rank-2 vector bundle over $\Lightcone$ also has vanishing first Chern number, $\Xi$ is isomorphic to the trivial bundle.
\end{proof}
Thus, the $R$ and $L$ gravitons are topologically nontrivial while the total graviton bundle is trivial. It is interesting to find that the $R$ and $L$ gravitons have different topology than the $R$ and $L$ photons, as the former have Chern number $\mp 4$ while the latter have Chern numbers $\mp 2$ \cite{PalmerducaQin_PT}. As the Chern number can be interpreted as a measure of how nontrivial or ``twisted'' a line bundle is \cite{Bott2013}, this shows that gravitons are more twisted than the photons. It is clear that photons and gravitons have different geometry; this is reflected in their helicities $\pm 1$ and $\pm 2$, respectively. However, it is important to note that a priori topology and geometry are distinct and reflect different mathematical structures, in this case the vector bundle structures and particular Poincar\'{e} symmetries of those bundles. There are certainly relationships between geometry and topology, but they take the form of deep results such as the Chern-Weil homomorphism and the Gauss-Bonnet theorem \cite{Tu2017differential}. We previously proved that vector bundle representations with the same helicity are topologically equivalent (\cite{PalmerducaQin_PT}, Theorem 35). Although we did not generally establish that vector bundles with different helicities are topologically distinct, the photon and graviton cases support this conjecture. Moreover, they suggest the general relationship $C_1 = -2h$.

We also remark that while the Berry connection was used to calculate $C_1$, there is no special relationship between the Berry connection and Chern numbers. Indeed, we could have used any connections on $\Xi$ and $\Xi_\pm$ (and there are infinitely many) in this calculation and obtained the same Chern numbers. The Hilbert-Schmidt product and its induced connection are geometric, whereas the Chern numbers are purely topological quantities and are independent of the various geometries one might impose on the vector bundle. 

Before proceeding, we note an alternative, faster method of establishing the topological triviality of $\Xi$. The key observation is that $\Xi$ is the complexification of a real vector bundle $\Xi^\Real$. Indeed, the transverse gauge condition (\ref{eq:transverse_gauge}) can be applied to the real vector space $V^\Real_3$ of symmetric traceless rank-2 tensors on $\Real^3$, generating a rank-2 real vector bundle $\Xi^\Real$. It is easy to see that $\Xi$ is the complexification of $\Xi^\Real$, that is, 
\begin{equation}
    \Xi = \Xi^\Real \otimes_\Real \Comp.
\end{equation}
The odd Chern numbers of complexifications of real vector bundles vanish \cite{Milnor1974,PalmerducaQin_PT}, thus
\begin{equation}
    C_1(\Xi) = 0.
\end{equation}
Again, since $C_1$ fully classifies rank-2 complex vector bundles over $\Lightcone$, this proves that $\Xi$ is a trivial bundle. The same method was used to establish the triviality of the photon bundle \cite{PalmerducaQin_PT}.

%%%%%%%%%%%%%%%%% Construction of Global Basis %%%%%%%%%%%%%%%%%%%%%%%%%%
\section{Construction of a global basis for gravitons}\label{sec:global_basis}
By Theorem \ref{thm:graviton_topology}, $\Xi$ is topologically trivial and thus there exists a smooth global basis for gravitons. However, neither of the proofs given in the previous section suggest a method to construct such a basis. In this section, we use the an adapted version of the clutching construction from algebraic topology to explicitly construct a globally smooth basis. The clutching construction is a general method of classifying vector bundles over any $n$-sphere $S^n$. It is based on the observation that although a sphere is topologically nontrivial, it can be split into two hemispheres which are trivial. Expositions of the clutching construction can be found in Refs. \cite{Atiyah_K_Theory, HatcherVBKT}. The modified method used here was developed by the authors to construct a global basis of photons \cite{PalmerducaQin_PT}.

We begin by splitting the two-sphere as $S^2 = D_U \cup D_L$ into the upper and lower closed hemispheres $D_U$ and $D_L$. Since $D_U$ and $D_L$ are contractible, the restricted bundles $\Xi|_{D_{U}}$ and $\Xi|_{D_{L}}$ are topologically trivial and therefore admit global bases. We construct such bases as follows. For each $\bfk = (\theta, \phi) \in D_U,$ define a rotation $R_{\bfk} \in \SO (3)$ which takes $\zhat$ to $\bfk$ via Rodrigues' rotation formula \cite{Dai2015}:
\begin{equation}\label{eq:Rodrigues}
    R_{\bfk} = I + \boldsymbol{K} \sin{\theta} + \boldsymbol{K}^2 (1-\cos{\theta})
\end{equation}
where
\begin{equation}
    \boldsymbol{K} = \begin{pmatrix}
        0 & 0 & \cos{\phi} \\
        0 & 0 & \sin{\phi}  \\
        -\cos{\phi} & -\sin{\phi} & 0
    \end{pmatrix}.
\end{equation}
Using $R_{\bfk}$, we can transport $A_{z,\pm}$ from $\zhat$ to all of $D_U$:
\begin{subequations}
\label{eq:A_U}
\begin{align}
    A_{U,1}(\bfk) &\doteq \Sigma(R_{\bfk})A_{z,+} \\
    A_{U,2}(\bfk) &\doteq \Sigma(R_{\bfk})A_{z,-}.
\end{align}
\end{subequations}
Since the action $\Sigma$ is unitary and the map $\bfk \mapsto R_{\bfk}$ from $D_U$ to $\SO(3)$ is smooth,
\begin{equation}
    \mathcal{F}_U = (A_{U,1} , A_{U,2})
\end{equation}
forms a smooth orthonormal basis of $\Xi|_{D_{U}}$. 
The $\SO(3)$ action $\Sigma$ defined by (\ref{eq:rotation_action}) naturally extends to an $\mathrm{O}(3)$ action. If we let $P = \mathrm{diag}(+1,+1,-1) \in \mathrm{O}(3)$ be the reflection across the $xy$ plane, then
\begin{subequations}
\label{eq:F_L}
\begin{align}
    \mathcal{F}_L &= (A_{L,1} , A_{L,2}) \\
    A_{L,1}(\bfk) &\doteq \Sigma(P)[A_{U,2}(P\bfk)] \\
    A_{L,2}(\bfk) &\doteq \Sigma(P)[A_{U,1}(P\bfk)]
\end{align}
\end{subequations}
is a smooth orthonormal basis of $\Xi|_{D_L}$; the order of the indices has been flipped for convenience. The bases $\mathcal{F}_{U,L}$ are smooth on their respective domains $D_{U,L}$, but they do not agree on the equator $S^1 \cong D_U \cap D_L$ where they overlap. Instead, for each $\bfk = (\cos(\phi),\sin(\phi),0)$ on the equator, there is a unitary transformation $T(\phi) \in \mathrm{U}(2)$ relating the two frames:
\begin{equation}
    \mathcal{F}_U(\phi) = T(\phi)\mathcal{F}_L(\phi).
\end{equation}
Direct calculation shows that
\begin{equation}
    T(\phi) = \begin{pmatrix}
        e^{4i\phi} & 0 \\
        0 & e^{-4i \phi}
    \end{pmatrix} \in \SU(2).
\end{equation}
$T:S^1 \rightarrow \SU(2) \seq \mathrm{U}(2)$ defines a loop in $\mathrm{U}(2)$. The clutching construction shows that the homotopy type of this loop fully classifies the vector bundle $\Xi|_{S^2}$ (and therefore $\Xi)$. In particular, this bundle is trivial if and only if $T$ it is contractible to a point i.e. if $T$ is nullhomotopic \cite{Atiyah_K_Theory, HatcherVBKT}. Since $T$ is actually a loop in the simply connected subgroup $\SU(2)$ of $\mathrm{U}(2)$, it is necessarily nullhomotopic. We thus obtain a third proof that $\Xi$ is topologically trivial. The advantage of this approach is that it can be adapted to give an explicit construction of a global basis for $\Xi$. The strategy is to extend the domain of $T$ from the equator $S^1$ to all of $D_L$. We do this in two steps.

%%%%%%%%%%%%%%%%%% Subsection: Homotopy from T to I %%%%%%%%%%%%%%%%%%%
\subsection{Homotopy from $T$ to the identity}
Since $T$ is nullhomotopic, there exists a homotopy $\tilde{T}:[0,1] \times S^1 \rightarrow \SU(2)$ between
\begin{equation}
T(\phi) = 
    \begin{pmatrix}
        e^{4i\phi} & 0 \\
        0 & e^{-4i \phi}
    \end{pmatrix}
    \;\;\; \text{ and } \;\;\;
    \mathds{1} = \begin{pmatrix}
    1 && 0 \\
    0 && 1
    \end{pmatrix}.
\end{equation}
The map
\begin{equation}
    (x,y,z) \in S^2 \mapsto \begin{pmatrix}
    x+iy && z \\
    -z && x-iy
    \end{pmatrix} \in \mathrm{SU}(2)
\end{equation}
gives an embedding of $S^2$ in $\SU(2)$. Under this embedding, $T(\phi)$ corresponds to the loop around the equator traversed clockwise four times, while $\mathds{1}$ corresponds to the constant loop at the point $(1,0,0)$. A homotopy can be obtained by deforming the equator into $(1,0,0)$ as depicted in Fig. \ref{fig:homotopy}. In particular, we set
\begin{subequations}
\begin{align}
    x(t,\phi) &= \cos(4\phi) \cos^2\Big(\frac{\pi t}{2}\Big) + \sin^2 \Big(\frac{\pi t}{2} \Big) \label{eq:x}\\
    y(t,\phi) &= \sin(4\phi)\cos\Big(\frac{\pi t}{2}\Big) \label{eq:y}\\
    z(t,\phi) &= \sin^2 (2\phi) \, \sin(\pi t), \label{eq:z}
\end{align}
\end{subequations}
and define the homotopy $\tilde{T}:[0,1] \times S^1 \rightarrow \SU(2)$ via 
\begin{equation}\label{eq:T_homo}
    \tilde{T}(t,\phi) = \begin{pmatrix}
        x(t,\phi) + iy(t,\phi) && z(t,\phi) \\
        -z(t,\phi) && x(t,\phi) - iy(t,\phi)
    \end{pmatrix}.
\end{equation}
It will be more convenient to change the homotopy parameter from $t \in [0,1]$ to $\theta \in [\frac{\pi}{2}, \pi]$:
\begin{equation}\label{eq:T_theta}
    \bar{T}(\theta,\phi) = \begin{pmatrix}
        \bar{x}(\theta,\phi) + i\bar{y}(\theta,\phi) && \bar{z}(\theta,\phi) \\
        -\bar{z}(\theta,\phi) && \bar{x}(\theta,\phi) - i\bar{y}(\theta,\phi)
    \end{pmatrix}
\end{equation}
where
\begin{subequations}
\begin{align}
    \bar{x}(\theta,\phi) &= \cos(4\phi)\sin^2(\theta) + \cos^2(\theta), \\
    \bar{y}(\theta,\phi) &= \sin(4\phi)\sin(\theta), \\
    \bar{z}(\theta,\phi) &= -\sin^2(2\phi)\sin(2\theta).
\end{align}
\end{subequations}

\begin{figure}
    \centering
    \includegraphics[width=8.6cm]{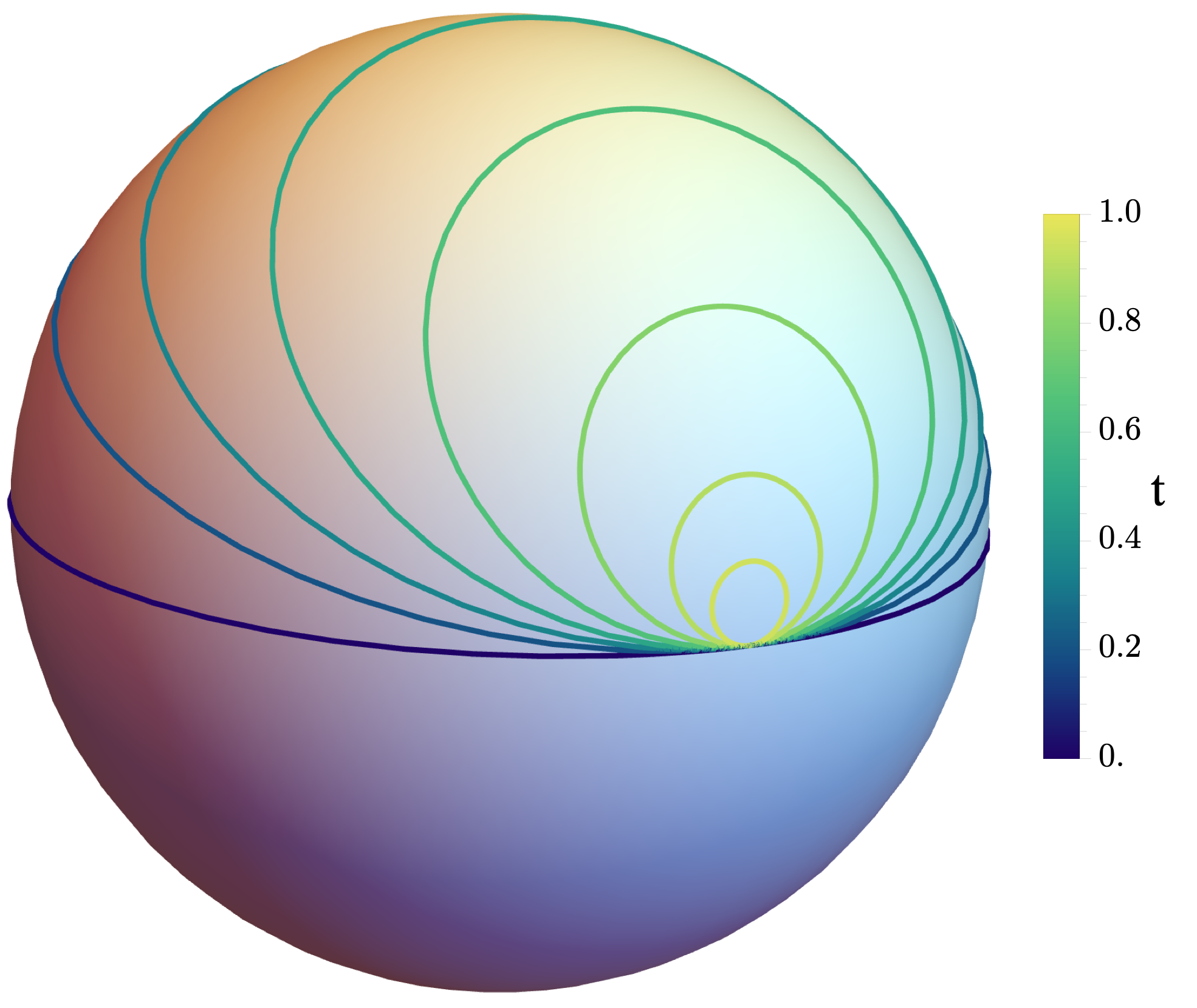}
    \caption{Illustration of the homotopy $\tilde{T}(t,\phi):[0,1]\times S^1 \rightarrow \mathrm{SU}(2)$ from $T(\phi)$ to the identity $\mathds{1}$. The image of $\tilde{T}$ resides on an embedding of $S^2$ in $\mathrm{SU}(2)$. The equator at $t=0$ corresponds to $T(\phi)$ and shrinks down to the identity $(1,0,0)$ at $t=1$.}
    \label{fig:homotopy}
\end{figure}
%%%%%%%%%%%%%% Construction of Smooth Frame of Xi %%%%%%%%%%%%%%%%%%%%%%%$

\subsection{Construction of a smooth frame of $\Xi|_{S^2}$}
Using the homotopy $\bar{T}$, we can patch together $\mathcal{F}_U$ and $\mathcal{F}_L$ into a continuous global basis $\mathcal{F}$ for $\Xi$:
\begin{equation}
    \mathcal{F}(\theta, \phi) = \begin{cases}
        \mathcal{F}_U(\theta, \phi) & \text{if } 0 \leq \theta < \frac{\pi}{2} \\
        \bar{T}(\theta, \phi) \mathcal{F}_L(\theta, \phi) & \text{if } \frac{\pi}{2} \leq \theta \leq \pi.
    \end{cases}
\end{equation}
The condition 
\begin{equation}\label{eq:homotopy_start_point}
    \bar{T}(\pi /2 , \phi) = T(\phi)
\end{equation}
ensures that $\mathcal{F}(\theta, \phi)$ is continuous along the equator, while $\bar{T}(\pi,\phi) = \mathds{1}$ ensures that no singularity occurs at the south pole. While $\mathcal{F}$ is continuous, it is not smooth since its $\theta$ derivatives jump abruptly at the equator. The last step is to use a smooth step function to smooth the frame at the equator.

To this end, let $g\in C^{\infty}(\mathbb{R})$ be any smooth monotonic function such that 
\begin{subequations}
\begin{gather}
        g(\theta \leq \pi/2)=\pi/2 \\
        g(\theta \geq \pi)=\pi \\
        g^{(n)}(\pi/2)=g^{(n)}(\pi)=0 \label{eq:g_deriv_0}
\end{gather}
\end{subequations}
for all $n\geq 1$. For concreteness, define
\begin{equation}
    h(t)= \begin{cases}
        e^{-1/t} & t > 0 \\
        0 & t \leq 0,
    \end{cases}
\end{equation}
and then set
\begin{equation}
    g(\theta) = \frac{\pi}{2}\Big(1 + \frac{h(\theta - \pi/2)}{h(\theta - \pi/2) + h(\pi - \theta)} \Big).
\end{equation}
We then define the frame
\begin{equation}\label{eq:smoothed_frame}
    \mathcal{F}_s(\theta, \phi) = \begin{cases}
        \mathcal{F}_U(\theta, \phi) & \text{if } 0 \leq \theta < \frac{\pi}{2} \\
        \bar{T}(g(\theta), \phi) \mathcal{F}_L(\theta, \phi) & \text{if } \frac{\pi}{2} \leq \theta \leq \pi
        \end{cases}
\end{equation}
which we claim is globally smooth. Indeed, it is obvious that $\mathcal{F}_s$ is smooth for $\theta \neq \frac{\pi}{2}$. It remains to show that
\begin{equation}
    \partial^m_\phi \partial^n_\theta \mathcal{F}_s\big|_{\theta = \frac{\pi}{2}^-} = \partial^m_\phi \partial^n_\theta \mathcal{F}_s\big|_{\theta = \frac{\pi}{2}^+}
\end{equation}
for all $m,n$. By (\ref{eq:homotopy_start_point}), (\ref{eq:g_deriv_0}), and (\ref{eq:smoothed_frame}), this is equivalent to the claim that
\begin{equation}\label{eq:smoothness_equiv}
    \partial^m_\phi \partial^n_\theta \mathcal{F}_U\big|_{\theta = \frac{\pi}{2}^-} = \partial^m_\phi \partial^n_\theta \Big[T(\phi)\mathcal{F}_L(\theta, \phi)\Big]_{\theta = \frac{\pi}{2}^+}.
\end{equation}
We now note that the equations (\ref{eq:A_U}-\ref{eq:F_L}) defining $\mathcal{F}_U$ and $\mathcal{F}_L$ smoothly extend to $0 < \theta < \pi$. It follows from direct calculation that
\begin{equation}
    \mathcal{F}_U(\theta, \phi) = T(\phi)\mathcal{F}_L(\theta, \phi)
\end{equation}
whenever both sides are defined, that is, when $\theta \neq 0, \pi$. This implies that (\ref{eq:smoothness_equiv}) holds, proving that $\mathcal{F}_s$ is smooth.

Although it is frequently assumed that such a global basis for gravitons exists \cite{Yu1999, Maggiore2005, Hu2021, Carney2024}, it appears this is the first explicit construction of such a basis. It allows an explicit decomposition of $\Xi$
\begin{equation}\label{eq:trivial_splitting_2}
    \Xi = \tau_1 \oplus \tau_2
\end{equation}
where $\tau_1$ and $\tau_2$ are the trivial line bundles over $\Lightcone$ generated by the basis $\mathcal{F}_s$.

We note that while the trivial splitting (\ref{eq:trivial_splitting_2}) may be useful in practical applications such as modeling gravitational waves, the subbundles $\tau_1$ and $\tau_2$ are not Poincar\'{e} symmetric and thus do not represent elementary particles. In other words, the splittings $\Xi = \tau_1 \oplus \tau_2$ and $\Xi = \Xi_+ \oplus \Xi_-$ both make sense topologically, but only the latter respects the Poincar\'{e} geometry of weak gravity.

%%%%%%%%%%% The nonexistence of linear subbundles %%%%%%%%%%%%%%%%%
\section{The nonexistence of linearly polarized subbundles}\label{sec:linear_polarizations}
Unlike the basis $\mathcal{B} = (B_+,B_\times)$ from equation (\ref{eq:B_plus_cross_Maggiore}), the basis $\mathcal{F}_s$ is globally smooth. However, $\mathcal{F}_s$ is significantly more complicated than $\mathcal{B}$. Indeed, the simplest types of polarizations are circular polarizations and linear polarizations. Circularly polarized vectors are those of the form $c(A_1 \pm i A_2)$ where $c \in \Comp$ and $A_1$ and $A_2$ are real and orthonormal to each other. While $\Xi$ uniquely decomposes into the sum of circularly polarized subbundles $\Xi_{\pm}$, these bundles are nontrivial and thus it is not possible to construct a globally smooth basis of circularly polarized gravitons. The basis $\mathcal{B}$ is an attempt to construct a linearly polarized basis, that is, a basis of real tensors; however, we showed that $\mathcal{B}$ necessarily has singularities and is not globally smooth. While this basis fails, it is natural to ask if it is possible to construct some other linearly polarized basis. $\mathcal{F}_s$ does not satisfy this criteria---while it is linearly polarized in the upper hemisphere, it is generally elliptically polarized in the southern hemisphere. In this section, we prove that it is impossible to construct a global linearly polarized basis for gravitons, but for a reason which is fundamentally different than the obstruction precluding a circularly polarized basis. While the bundles $\Xi_\pm$ are topologically nontrivial, they are still well-defined topological objects. In contrast, $\Xi$ possesses no linearly polarized subbundles, even topologically nontrivial ones.

To prove this, we must first be more precise about the definition of linearly polarized vectors and subbundles. Indeed,  if $A_0$ is real than it is certainly linearly polarized, however, $cA_0$ should also be considered linearly  polarized for any $c \in \Comp$ since $cA$ differs from $A_0$ only in a phase and scaling. Recall that $\Xi\cong \Xi^\Real \otimes_\Real \Comp$ is the complexification of the underlying real vector bundle $\Xi^\Real$. Recall also that in a tensor product $X \otimes Y$, simple tensors are those of the form $x \otimes y$ for $x \in X$ and $y \in Y$; general elements of $X \otimes Y$ are finite sums of simple tensors \cite{dummit2003}.
\begin{definition}[Linearly polarized vectors]
    $(k,A) \in \Xi \cong \Xi^\Real \otimes_\Real \Comp$ is linearly polarized if A is a simple tensor, that is, if $A = cA_0 = A_0 \otimes_\Real c$ for some $c \in \Comp$ and $A_0 \in \Xi^\Real$. Equivalently, $(k,A)$ is linearly polarized if the set $\{\rp(cA):c \in \Comp\}$ is a one-dimensional real vector space.
\end{definition}

\begin{definition}[Linearly polarized subbundles]
    A line bundle $\ell$ is a linearly polarized subbundle of $\Xi$ if every element of $\ell$ is linearly polarized. Equivalently, $\ell$ is linearly polarized if it is the complexification of a real line subbundle $\ell^\Real$ of $\Xi^\Real$.
\end{definition}
The main result of this section is the following theorem.
\begin{theorem}\label{Thm:no_linear_polarizations}
    $\Xi$ has no linearly polarized subbundles. As a result, there exists no smooth, linearly polarized global extension of the basis $(B_{z,+} ,B_{z,\times})$ at $\zhat$.
\end{theorem}
To prove this, we first establish the following lemma. In it, we use the Euler class $e$, which is a characteristic class for real, oriented vector bundles. The Euler class is less commonly used in physics than the Chern classes; expositions can be found in Refs. \cite{Bott2013,HatcherVBKT}.
\begin{lemma}\label{lm:linear_lemma}
    If $\Xi$ has a linearly polarized subbundle then the Euler class of $\Xi^\Real$ vanishes.
\end{lemma}
\begin{proof}[Proof of Lemma \ref{lm:linear_lemma}]
    Assume $\Xi$ has a linearly polarized subbundle, and therefore that $\Xi^\Real$ possesses a real line subbundle. Since $S^2$ is simply connected, there exist no nontrivial real line bundles over $S^2$ (\cite{PalmerducaQin_PT}, Lemma 21). Thus, $\Xi^\Real$ has a trivial line subbundle and must therefore possess a globally nonvanishing section. $\Xi^\Real$ is orientable since every real vector bundle over $S^2$ is orientable (\cite{Bott2013}, Proposition 11.5). If a real oriented vector bundle has a nonvanishing section, then its Euler class is zero (\cite{HatcherVBKT}, Proposition 3.13(e)). Thus, the Euler class of $\Xi^\Real$ vanishes.
\end{proof}
\begin{proof}[Proof of Theorem \ref{Thm:no_linear_polarizations}]
    By the lemma it suffices to show that the Euler class of $\Xi^\Real$ is nonzero. The tensors $B_{z,+}$ and $B_{z_\times}$ form an orthonormal basis of $\Xi^{\Real}(\zhat)$. Then
    \begin{subequations}
    \begin{align}
        \tilde{\mathcal{B}}(\bfk) &= (B_+(\bfk),B_\times(\bfk)), \\
        B_{+}(\bfk) &\doteq \Sigma\Big(R_{\zhat}(\phi)R_{\xhat}(\theta)\Big)B_{z,+} \\
        B_{\times}(\bfk) &\doteq \Sigma\Big(R_{\zhat}(\phi)R_{\xhat}(\theta)\Big)B_{z,\times}
    \end{align}
    \end{subequations}
    forms a smooth orthonormal basis of $\Xi_{\Real}(\theta, \phi)$ at all $\bfk$ except at the north and south poles, where 
    $$\bfk = |\bfk|(\sin \theta \cos \phi, \sin \theta \sin \phi, \cos \theta).$$ We can then calculate the corresponding Berry connection matrix with respect to the spherical coordinate system for $\Lightcone$:
    \begin{align}
        \omega_{\Real} &= \langle \tilde{\mathcal{B}},d \tilde{\mathcal{B}}\rangle \\
        &=\begin{pmatrix}
            0 & -2\cos(\theta)d\phi \\
            2\cos(\theta)d\phi & 0
        \end{pmatrix}.
    \end{align}
    The curvature matrix in these coordinates is
    \begin{align}
        \Omega_\Real &= d\omega_\Real + \omega_\Real \wedge \omega_\Real \\
        &= \begin{pmatrix}
            0 & 2dS \\
            -2dS & 0
        \end{pmatrix}.
    \end{align}
    The Pfaffian of this matrix is \cite{Tu2017differential}
    \begin{equation}
        \mathrm{Pf}(\Omega_\Real) = 2dS.
    \end{equation}
    The Euler class is the de Rham cohomology class of the Pfaffian \cite{Tu2017differential}, so
    \begin{equation}
        e(\Omega_\Real) = 2[dS]
    \end{equation}
    where the brackets denote the cohomology class. $dS$ is not exact (since $\int_{S^2} dS = 4\pi \neq 0)$, so $e(\Omega_\Real)$ is nonvanishing, and therefore $\Xi$ has no linear subbundles.
\end{proof}

%%%%%%%%%%% Gravitons do not hvae SAM or OAM %%%%%%%%%%%%%%%%%%%%%%%%%
\section{Gravitons do not have spin or orbital angular momentum}\label{sec:OAM_SAM}
An important consequence of the nontrivial topology of $R$ and $L$ gravitons is that their angular momentum cannot be split into SAM and OAM parts. We begin by reviewing the more familiar case of photon angular momentum. There has been a long controversy over the legitimacy of splitting photon angular momentum into SAM and OAM parts \cite{Akhiezer1965, VanEnk1994,Bliokh2010, Bialynicki-Birula2011, Leader2013, Bliokh2014, Leader2016, Leader2019, PalmerducaQin_PT}. Although experimentalists have found the concept of the OAM of light very useful (see Ref. \cite{Bliokh2014}, and references therein), it is generally believed that there is no gauge invariant way to split photon angular momentum into SAM and OAM, although even this point has been debated \cite{Chen2008, Leader2014}. Even more disquieting is a conclusion reached by Leader and Lorc\'{e} in their review article \cite{Leader2014} that ``the choice of a particular decomposition is essentially a matter of taste and convenience.'' The issue with this viewpoint is that it fails to account for the fundamental defining feature of angular momentum, namely that angular momentum must derive from an $\SO(3)$ symmetry of the system. That is, any angular momentum operators, such as the SAM and OAM operators, must satisfy $\so(3)$ commutation relations so that they are generators of an $\SO(3)$ symmetry \cite{VanEnk1994,PalmerducaQin_PT}; this was acknowledged by Leader and Lorc\'{e} in a corrigendum \cite{Leader2019} to the aforementioned review article. Furthermore, SAM operators must generate an $\SO(3)$ symmetry of the internal (fiber) DOFs. The issue is that none of the proposed SAM and OAM operators are both well-defined and generators of $\SO(3)$ symmetries. The most obvious attempt to define the SAM operator is to take it to be the conventional spin 1 operator \cite{Akhiezer1965}. However, this operator is ill-defined on photons because the spin $s$ operator must act on an internal $2s+1$-dimensional vector space. Thus, the spin 1 operators must act on a three-dimensional vector space while the internal polarization space of photons is only two-dimensional; the result is that the spin 1 operator is not well-defined on the space of photons. This is concretely illustrated by the fact that spin 1 operators do not preserve the transversality of photon polarizations, and thus map photons to nonphysical modes \cite{VanEnk1994,Bialynicki-Birula2011}. This issue is well-known, and based on a number of different arguments, it has been proposed that the SAM and OAM operators should instead be defined via \cite{VanEnk1994,Bliokh2010,Bialynicki-Birula2011}
\begin{subequations}
\begin{gather}
    \boldsymbol{J}_s = (\hat{\bfk} \cdot \boldsymbol{J})\hat{\bfk} \\
    \boldsymbol{J}_o = \boldsymbol{J} - \boldsymbol{J}_s,
\end{gather}
\end{subequations}
where $\boldsymbol{J}$ is the total angular momentum operator. While these vector operators are well-defined, they do not satisfy $\so(3)$ commutation relations, and therefore do not generate a 3D rotational symmetry, which is the defining feature of angular momentum \cite{VanEnk1994, Leader2019, PalmerducaQin_PT}. Instead, they satisfy the nonstandard commutation relations
\begin{subequations}
\begin{align}
    [J_{s,a},J_{s,b}]&=0 \label{eq:Lie_1},\\ 
    [J_{o,a},J_{s,b}]&= i \epsilon_{abc}J_{s,c} \label{eq:Lie_2},\\
    [J_{o,a},J_{o,b}]&= i \epsilon_{abc}(J_{o,c}-J_{s,c}) \label{eq:Lie_3}
\end{align}
\end{subequations}
where the indices $a,b,c$ run over $\{1,2,3\}$. It can be shown that $\boldsymbol{J}_s$ is associated with an $\Real^3$ Lie algebra symmetry (rather than an $\mathfrak{so}(3)$ symmetry), while $\boldsymbol{J}_o$ does not generate any symmetry at all \cite{PalmerducaQin_PT}. It is then clear that these are not truly angular momentum operators.

Despite these foundational issues, the so-called OAM of light has been a popular concept in research \cite{Shen2019,Franke2022}. Given this popularity as well as the breakthroughs in the detection of gravitational waves of the last decade \cite{GravWave2016Blackhole,GravWave2017NeutronStar}, it is unsurprising that researchers have begun looking into the ``OAM'' of gravitational waves \cite{Bialynicki-Birula_2016, Baral2020, Wu2023}. However, we show now that, as in the case of the angular momentum of light, gravitational angular momentum cannot be split into SAM and OAM. 

The (total) angular momentum operator $\boldsymbol{J}$ of gravitons is the generator of the Poincar\'{e} action $\Sigma|_{\SO(3)}$ restricted to $\SO(3) \seq \poincare$. If the angular momentum splits into SAM and OAM, then
\begin{equation}\label{eq:LS_decomp}
    \boldsymbol{J} = \boldsymbol{L} + \boldsymbol{S}
\end{equation}
where $\boldsymbol{L}$ and $\boldsymbol{S}$ satisfy $\so(3)$ commutation relations, and thus generate $\SO(3)$ actions $\Sigma_L$ and $\Sigma_S$ on $\Xi$. Furthermore, because spin symmetries are internal symmetries, $\boldsymbol{S}$ must generate an internal symmetry, that is, the action of $\Sigma_S$ on $(k,A)$ does not change $k$. More formally, we say that $\Sigma_S$ is stabilizing.
\begin{theorem}
    There exist no nontrivial SAM-OAM splittings of the angular momentum of gravitons as in equation (\ref{eq:LS_decomp}).
\end{theorem}
\begin{proof}
    Suppose there were such a splitting. Since $\Sigma_S$ is stabilizing, each fiber $\Xi(\bfk)$ is a two-dimensional vector space representation of $\SO(3)$ with the action $\Sigma_S$. However, the (nonprojective) irreducible representations of $\SO(3)$ are labeled by their integer spin $s$ and have dimensions $2s+1$ \cite{Hall2013}. Thus, the representation $\Sigma_S$ on $\Xi(\bfk)$ must be the direct sum of two spin $0$ representations. Since the spin $0$ representation is simply the trivial representation which leaves the vectors unchanged, $\Sigma_S$ must be the identity operator. Thus, $\boldsymbol{S} = 0$, and therefore $\boldsymbol{J} = \boldsymbol{L}$ and $\Sigma|_{\SO(3)} = \Sigma_L$. Thus, there is no nontrivial way to split gravitational angular momentum into SAM and OAM parts.
\end{proof}

This argument against the splitting of gravitational OAM and SAM appears geometric, relying on the rotational symmetry of the gravitons. However, as in the case of photon angular momentum, this obstruction is, at its root, topological. As the discussion is essentially the same as for the photon case and the technical apparatus somewhat large, we summarize the discussion here and refer the reader to Ref. \cite{PalmerducaQin_PT} for the detailed mathematics. Elementary particles can be described as vector bundle representations of the Poincar\'{e} group $\ISO^+(3,1)$ \cite{Simms1968,Asorey1985,PalmerducaQin_PT}.  The little group of particles with four-momentum $k$ is defined as the subset of the Lorentz group $\SO^+(3,1)$ which do not change $k$. Up to group isomorphism, the little group $H$ is independent of $k$. By its definition, the little group depends only the momentum space $M$ and not on the fibers (the polarizations). Furthermore, each particle has a canonical little group action which is stabilizing in the sense that the action does not change the fiber. Thus, the little group describes an internal symmetry of the particles. For massive particles, the momentum space is a mass hyperboloid and the little group is $\SO(3)$. The canonical little group action of massive particles thus describes an internal $\SO(3)$ symmetry; its generator is the spin angular momentum. However, a topological singularity occurs as the $m\rightarrow 0$ limit is realized. The momentum space jumps from the topologically trivial mass hyperboloid to the nontrivial forward lightcone. This topological singularity allows the little group to jump as well from $\SO(3)$ to $\ISO(2)$, where the latter is the 2D Euclidean group consisting of translations and rotations of $\Real^2$. However, the translations always act trivially, so the little group is effectively just $\SO(2)$. The generator of this two-dimensional rotational symmetry is the helicity operator $\hat{\bfk} \cdot \boldsymbol{J}$. Since this is not an $\SO(3)$ generator, helicity is not an angular momentum operator. Indeed, helicity is a scalar operator, and so it is not even a vector operator. The commonly proposed SAM operator $\boldsymbol{J}_s = (\hat{\bfk} \cdot \boldsymbol{J}) \hat{\bfk}$ can be seen as an ad hoc attempt to ``vectorize'' the helicity operator. However, this attempt fails because the components of this new operator commute, and thus describe an $\Real^3$ symmetry rather than an $\SO(3)$ symmetry. We can summarize the situation by saying that massive particles have spin while massless particles have helicity. Something special happens for massive particles which does not happen for massless particles, namely, the little group happens to coincide with $\SO(3)$, so spin is actually an angular momentum. The same is not true of helicity.

Heuristically, it is not surprising that the topological nontriviality of R and L photons and gravitons obstructs splitting the angular momentum into SAM and OAM. An SAM-OAM decomposition relies on nicely splitting the internal (fiber) and external (base manifold) DOFs, and likewise splitting the rotational symmetry into parts acting internally and externally. Furthermore, the action $\Sigma$ acts independently on $\Xi_R$ and $\Xi_L$, so the same would be true of the spin and orbital symmetries $\Sigma_S$ and $\Sigma_L$. However, the topological nontriviality of the R and L gravitons means that their internal and external DOFs are twisted together, and this precludes splitting the rotational action into internal and external parts.

We finish this section by remarking on a potential point of confusion. There are known methods for starting with a spin $s \geq 1$ representation and subsequently applying gauge constraints to eliminate DOFs until one ends up with only two DOFs which correspond to massless helicity $\pm s$ particles \cite{Maggiore2005,Tong2006}. This might appear to suggest that these massless particles have spin $s$. However, while these techniques highlight an interesting relationship between massive spin $s$ particles and massless helicity $\pm s$ particles, in eliminating DOFs one also destroys the internal $\SO(3)$ symmetry of the system, and thus it is incorrect to say particles with helicity $\pm s$ have spin $s$. Indeed, the representations of $\SO(3)$ are completely classified \cite{Hall2013}, and there is a unique (nonprojective) vector space representation of $\SO(3)$ on the two-dimensional fibers of a massless particle, namely, the trivial representation which is the direct sum of two spin $0$ representations. Thus, the fibers of the massless particle do not form spin $s$ representations.

%%%%%%%%%%%%% Curved Spacetimes %%%%%%%%%%%%%%%%%%%%%%
\section{Topology of gravitational waves in curved spacetimes} \label{sec:curved_spacetime}
The above results were derived assuming a flat spacetime background as in Eq. (\ref{eq:flat_assumption}). In this section we discuss the situations in which this topological analysis extends to gravitational waves in curved spacetimes.

It is not generally possible to unambiguously identify a gravitational wave component of a solution of the Einstein field equations, even when assuming the gravitational waves are small amplitude \cite{Maggiore_grav_2007}. One fairly general situation in which gravitational waves can be isolated in the geometric optics limit, also known as the short-wave approximation, in which there is separation of spatial and/or temporal scales. We will follow the treatments of Misnor, Thorne, and Wheeler \cite{Misnor1973} and Maggiore \cite{Maggiore_grav_2007}, expressing the metric as
\begin{equation}\label{eq:WKB_metric}
    g_{\mu \nu}(x) = \bar{g}_{\mu \nu}(x) + h_{\mu \nu}(x)
\end{equation}
where $\bar{g}_{\mu \nu}$ is the background metric and $h_{\mu \nu}$ describes small amplitude gravitational waves with $g_{\mu \nu} \sim 1$ and $|h_{\mu \nu }| \sim \delta \ll 1$. The short-wave approximation is applicable when $h_{\mu \nu}$ varies on length or time scales $\lambda$ much shorter than the characteristic length or time scales $\bar{L}$ of $\bar{g}_{\mu \nu}$: $|\lambda / \bar{L}| \sim \epsilon \ll 1$. If the stress energy tensor $T_{\mu \nu}$ vanishes to leading order, then $\epsilon \sim \delta$, but in general the relative scaling of $\epsilon$ and $\delta$ may be different \cite{Maggiore_grav_2007}. Defining $h = \tensor{h}{^{\alpha}_{\alpha}}$ (where indices are raised and lowered with $\bar{g}_{\mu \nu}$) and the traceless perturbation
\begin{equation}
    \bar{h}_{\mu\nu}(x) = h_{\mu \nu} - \frac{1}{2}\bar{g}_{\mu \nu}h,
\end{equation}
we make the eikonal ansatz
\begin{equation}
    \bar{h}^{\mu \nu} (x) = [A(x) e^{\mu \nu} (x) + \epsilon B^{\mu \nu}(x) + \ldots ] e^{i\theta(x)/\epsilon}
\end{equation}
where the terms inside the bracket are slowly varying and $e^{\mu \nu}$ is traceless and symmetric. Applying the standard geometric optics formalism with $k^\mu = \partial_\mu \theta$, the fast dynamics are, to leading order, decoupled from the slow dynamics and not explicitly dependent on $T^{\mu \nu}$. One finds the rays are null geodesics,
\begin{gather}
    k_\alpha k^\alpha = 0 \\
    k_{\alpha| \beta}k^\beta = 0,
\end{gather}
(where the vertical bar denotes the covariant derivative corresponding to $\bar{g}_{\mu \nu}$) and the polarization vectors are transverse and parallel transported along the geodesics
\begin{align}
    e^{\mu \nu}k_\nu = 0 \\
    \tensor{e}{^{\mu \nu}_{| \alpha}} k^\alpha = 0
\end{align}
We focus locally on  waves propagating from a fixed position which we choose to be $x=0$. In analogy with the temporal gauge condition (\ref{eq:complete_gauge}), it is always possible to make a gauge transformation such that $e^{\mu 0} = 0$ over scales of order $\lambda$ (\cite{Misnor1973}, Eq. (35.69)). Thus, the gauge can be fixed such that the polarization vectors of locally propagating gravitational waves in a curved spacetime satisfy nearly the same transverse-traceless gauge conditions as those of a flat spacetime (see Eq. (\ref{eq:gauge_constraints})). The one difference is that $e^{ij}$ is traceless with respect to $\bar{g}_{\mu \nu}$ rather than $\eta_{\mu \nu}$ as in the vacuum case:
\begin{equation}\label{eq:curved_trace}
    e^{ij}\bar{g}_{ij} = 0.
\end{equation}
This can be remedied if from the outset we switch to locally inertial coordinates with respect to $\bar{g}_{\mu \nu}$ at $x=0$. That is, let $\tilde{x}^\mu$ be such that $\bar{g}_{\mu \nu}dx^\mu dx^\nu = [\eta_{\mu \nu} + O(x^2) ]d\tilde{x}^\mu d\tilde{x}^\nu$. This can be accomplished with a transformation of the form
\begin{equation}
    \tilde{x}^\mu = \tensor{K}{^{\mu}_\nu}x^\nu + \frac{1}{2}\tensor{L}{^{\mu}_{\nu \sigma}}x^\nu x^\sigma + O(x^3)
\end{equation}
where $\tensor{K}{_{\mu\nu}} \sim |\tensor{\bar{g}}{_{\mu\nu}}(0)|^{1/2} \sim 1$ and $|L_{\mu \nu \sigma}| \sim |\bar{g}_{\mu \nu, \sigma}(0)| \sim \bar{L}^{-1}$. Writing Eq. (\ref{eq:WKB_metric}) in the $\tilde{x}$ coordinate system, for $\frac{\tilde{x}}{L} \sim \epsilon$ we have, 
\begin{equation}
    \tilde{g}_{\mu \nu} = [\eta_{\mu \nu} + O(\epsilon^2)]_\text{slow} + [h_{\alpha \beta}\tensor{K}{^\alpha _\mu}\tensor{K}{^\beta _\nu} + O(\epsilon \delta)]_\text{fast}.
\end{equation}
The brackets denote the slow and fast terms which have characteristic lengths $\bar{L}$ and $\lambda$, respectively. Thus, it is possible to choose coordinates such that the background metric is $\eta_{\mu \nu}$ to linear order while preserving the scale separation required by the short-wave approximation. In such coordinates Eq. (\ref{eq:curved_trace}) requires that $e^{ij}$ is a traceless matrix, and we see that the local polarization vectors $e^{ij}$ now satisfies precisely the same conditions as the vacuum transverse traceless gauge. Thus, within the geometric optics limit, the topology of gravitational waves in a curved background spacetime is the same as the topology of the vacuum case.

In the particularly simple case of a Friedmann–Lemaître–Robertson–Walker (FLRW) cosmology, we can go beyond the local analysis of geometric optics, and describe the global topology of small amplitude gravitational waves, which ends up being the same as the vacuum case. Here, gravitational waves are described by transverse-traceless perturbations $h_{ij}$ to the Robertson-Walker metric,
\begin{equation}
    ds^2 = a(\eta) \big[ -d\eta^2 + dx^2 + dy^2 + dz^2 + h_{ij} dx^i dx^j \big],
\end{equation}
where the scale factor $a(t)$ satisfies the Friedmann equations and $dt = a(t)d\eta$ \cite{Creighton2011}. Assuming no perturbations to the stress-energy tensor, the gravitational waves satisfy the damped wave equation
\begin{equation}
    \nabla^2 h_{ij} - \frac{\partial^2}{\partial \eta^2}h_{ij} - 2\frac{1}{a}\frac{da}{d\eta}\frac{\partial h_{ij}}{\partial \eta} = 0.
\end{equation}
The scale factor $a(t)$ can thus cause damping or amplification of gravitational waves as well as redshift, but it does not effect the plane wave mode structure: the topology of gravitational waves in an FLWR cosmology are precisely the same as those in vacuum.

Our formalism, which characterizes the topology of plane waves, does not immediately extend to analysis of the global quasinormal modes in simple black hole geometries, which are radially propagating \cite{Creighton2011}. Locally, high frequency gravitational waves in such spacetimes are described by geometric optics and thus have the topology of vacuum gravitons. However, additional work and a modified framework is needed to determine if such topology also describes the global quasinormal modes.

%%%%%%%%%%%%%%% Conclusion %%%%%%%%%%%%%%%%%%%%%%%%%%%
\section{Conclusion}

This study of gravitons gives a further illustration of the fact that there is a fundamental connection between massless particles and nontrivial topology. An immediate implication of the nontriviality of $\Xi_+$ and $\Xi_-$ is that it is not possible to construct global bases of R and L gravitons. It is also not possible to construct a global basis of linearly polarized gravitons, but this is due to the more fundamental issue that such polarizations are not even globally well-defined. In light of these facts, it is rather surprising that the full set of gravitons $\Xi$ is topologically trivial, allowing the construction of a smooth global basis. Such global bases have been previously invoked in the literature, although the explicit expressions given for them were not smooth \cite{Yu1999, Maggiore2005, Hu2021, Carney2024}.

The Poincar\'{e} symmetry of the graviton bundle induces a geometric decomposition of $\Xi$ into the topologically nontrivial R and L graviton bundles. These nontrivial bundles have geometry characterized by the helicities $\pm 2$ and topology characterized by the Chern numbers $\mp 4$. It is interesting to note that the $-2$ factor relating the helicities and Chern numbers also appears for the $R$ and $L$ photons which have helicities $\pm 1$ and Chern numbers $\mp 2$ \cite{Bliokh2015, PalmerducaQin_PT}. We conjecture that this relationship $C = -2h$ between topology and geometry generally holds for massless particles. Future work will examine the validity of this conjecture.

An important finding is that the topological nontriviality of $\Xi_+$ and $\Xi_-$ obstructs the splitting of graviton angular momentum into SAM and OAM. That is, gravitons, like photons, have well-defined angular momentum but it is not possible to decompose it into two parts depending only on the internal and external DOFs, respectively. Indeed, that the $R$ and $L$ gravitons have nontrivial topology means that they cannot simply be written as $\Lightcone \times \Comp$, showing that the internal and external DOFs do not split nicely. They are rather twisted in a nontrivial way, ultimately preventing the splitting of angular momentum into OAM and SAM. That massless particles can be topologically nontrivial is a recent finding, and we anticipate that future research will uncover more important physical consequences of such nontriviality.

\begin{acknowledgments}
This work is supported by U.S. Department of Energy (DE-AC02-09CH11466).
\end{acknowledgments}

\bibliography{gt}
\end{document}